\documentclass{llncs}
\usepackage{geometry}
\usepackage{graphicx,verbatim}
\usepackage{lineno}
\usepackage{amsmath,amssymb}
\usepackage{xspace}
\usepackage{color}
\usepackage{epsfig}
\usepackage[algosection,lined,ruled,linesnumbered]{algorithm2e}
\usepackage{subfigure}

\graphicspath{{./Fig/}}
\makeatletter
\newbox\ProofSym
\setbox\ProofSym=\hbox{%
\unitlength=0.18ex%
\begin{picture}(10,10)
\put(0,0){\framebox(9,9){}}
\put(0,3){\framebox(6,6){}}
\end{picture}}
\renewenvironment{proof}[1][Proof.]{\O@proof{#1}}{\O@endproof}
\def\O@proof#1{\O@ProofSymtrue\topsep\z@%\@topsepadd\smallskipamount%
  \partopsep\z@%
  \trivlist\@ifstar{\item[]}{\item[\hskip\labelsep\it #1 ]}}
\def\O@endproof{\ifO@ProofSym\relax\hfill\copy\ProofSym\linebreak
  \fi\endtrivlist}
\def\DisplayProofSym{\vskip-\lastskip\vskip-8pt
  \hbox to \hsize{\hfill\copy\ProofSym}\O@ProofSymfalse}
\newif\ifO@ProofSym

\renewcommand{\leq}{\leqslant}
\renewcommand{\geq}{\geqslant}

\title{ {Complexity results for $k$-domination and $\alpha$-domination problems and their variants}}
\author{
Davood Bakhshesh \and Mohammad Farshi\and Mahdieh Hasheminezhad }
\institute{Combinatorial and Geometric Algorithms Lab., Department of Computer Science,
 Yazd University. \email{dbakhshesh@gamil.com, mfarshi@yazd.ac.ir, hasheminezhad@yazd.ac.ir }}
\begin{document}
\maketitle
\begin{abstract}
\nonumber
Let $G=(V, E)$ be a simple and undirected  graph. For some integer $k\geq 1$, a set $D\subseteq V$ {is said to be} a {\it k-dominating set} in $G$ if every vertex $v$ of $G$ outside $D$ has at least $k$ {neighbors} in $D$. Furthermore, for some real number $\alpha$ with $0<\alpha\leq1$, a set $D\subseteq V$ is called {an} {\it $\alpha$-dominating set} in  $G$ if every vertex $v$ of $G$ outside $D$ has at least $\alpha\times d_v$ {neighbors} in $D$, where $d_v$ is the degree of $v$ in $G$. {The cardinality of a minimum $k$-dominating set and  a minimum $\alpha$-dominating set in  $G$ {is said to be} } {the} {\it $k$-domination number} and  {the} {\it $\alpha$-domination number} of $G$, respectively. In this paper, we present some approximability and inapproximability results on  the problem of finding of $k$-domination number and $\alpha$-domination number of some  classes of graphs.  Moreover, we  introduce a generalization of $\alpha$-dominating set {which we call} {{\it \mbox{$f$-dominating set}}}. Given a function  $f:\mathbb{N}\rightarrow \mathbb{R}$, where $\mathbb{N}=\{1, 2, 3, \ldots\}$, a set $D\subseteq V$ {is said to be} an  {$f$-dominating} set in $G$ if  every vertex $v$ of $G$ {outside $D$ has at least $f(d_v)$ neighbors in $D$.} {{We prove} NP-hardness of the problem of finding of a minimum $f$-dominating set in $G$, for a large family of functions $f$.}
\end{abstract}
\textbf{Keywords:}
 $f$-Domination, $\alpha$-domination, $k$-domination, approximation.
\\
\\
\textbf{2010 Mathematics Subject Classification:} 05C69, 68R05, 68Q25.
%\MSC[2010] 00-01\sep  99-00
\section{Introduction}
Let $G=(V, E)$ be an undirected and simple graph. A set $D\subseteq V$ is called {a} {\it dominating} set in $G$ if every vertex of $G$ outside $D$ has at least one neighbor in $D$, or equivalently $|N(v)\cap D|\geq1$, where $N(v)$ is the set of all neighbors of $v$ in  $G$. {The} {cardinality} of a minimum dominating set in $G$ is called {the} {\it domination number} of $G$ denoted by $\gamma(G)$. In  {the} past three decades, wide researches have been done on the domination number of graphs and related problems. For a survey of the area of domination in graphs and {its}  applications we refer the reader to \cite{haynes1998fundamentals1,haynes1998domination}.   In 1985, Fink and Jacobson \cite{fink1985n2,fink1985n1} introduced {the concept of a $k$-dominating set}. Let $k$ be a real number with $k\geq 1$. A set $D\subseteq V$ is called a  {\it $k$-dominating} set in $G$ if for every vertex $v$ outside $D$, $|N(v)\cap D|\geq k$. {The} {cardinality} of a $k$-minimum dominating set is called {the} {\it $k$-domination number} {of $G$ and denoted} by $\gamma_k(G)$. In 2000,  Dunbar  et al.  \cite{dunbar2000alpha},  introduced the concept  of $\alpha$-domination. Let $\alpha$ be a real number with $0<\alpha\leq1$. A set $D\subseteq V$ is called an $\alpha$-dominating set in $G$ if for every vertex $v$ of $G$ outside $D$,  $|N(v)\cap D|\geq \alpha\times d_v$, where $d_v:=|N(v)|$ is the degree of $v$. {The} {cardinality} of a minimum $\alpha$-dominating  set is called {the} {\it $\alpha$-domination number} of $G$ denoted by $\gamma_{\alpha}(G)$.

{In this paper, we show that for any integer $k\geq 1$, the problem of finding a minimum $k$-dominating set in a given graph of maximum degree $k + 2$ is APX-complete (that is, there is no PTAS for the problem unless P = NP).} Furthermore, we present some approximability and inapproximability results on  the problem of finding of $k$-domination number and $\alpha$-domination number of some  classes of graphs. Note that in this paper, we consider  {only  graphs with}   no isolated vertices. We can easily extend the results for the graphs with isolated vertices.

Another interesting problem that we consider in this paper {is that of  approximating} {the} $k$-domination number and $\alpha$-domination number of  $p$-claw free {graphs, which are graphs} without complete bipartite graph $K_{1,p}$ as induced subgraph. We propose an  approximation algorithm  for this problem.   We will show that our approximation algorithm has {an approximation ratio} better than previously known values when the maximum degree of input graph  {satisfies} some special conditions. 
% In 1993, Christoph Stracke and Lutz Volkmann \cite{stracke1993new} introduced {the} concept of {\it $f$-dominating} {sets}.  Suppose that $f$ is an integer-value {function on $V$}. A set $D\subseteq V$ is called an $f$-dominating set in $G$ if for every vertex $v$ outside $D$, $|N(v)\cap D|\geq f_{deg}(v)$. {The cardinality} of a minimum $f$-dominating set in $G$ is called the {\it $f$-domination number} of $G$ {and denoted by $\gamma_f_{deg}(G)$}.  Some researches has been done on $f$-domination \cite{chen1998upper,zhou1996f,  zhou2000inequalities, zhou2014invariants}.      
%Let  $G=(V,E)$ be a graph with vertex set $V:=\{v_1, v_2, \ldots, v_n\}$, and let $D_V=\{d_{v_1},d_{v_2},\ldots, d_{v_n}\}$ be the set of degrees of all vertices in $G$. 

Now consider the definition of  {$\alpha$-domination}. One generalization of this concept is that instead of having at least $\alpha\times d_v$ neighbors in $D$ for each vertex $v\not \in D$, we have at least $f(d_v)$  neighbors in $D$, for some  function $f$. By selecting $f(x)=\alpha x$,  {the definition matches that of $\alpha$-domination.}  Hence,   in this paper, {we define the  notion of an} {{\it $f$-dominating}} {set}. Given a function  $f:\mathbb{N}\rightarrow \mathbb{R}$, where $\mathbb{N}=\{1, 2, 3, \ldots\}$,  a set $D\subseteq V$ is called an  {$f$-dominating} set in $G$ if for every vertex $v$ of $G$ outside $D$, $|N(v)\cap D|\geq f(d_v)$. {The} {cardinality} of a minimum $f$-dominating set in a graph $G$ is called {the} {\it $f$-domination number} of $G$, and the problem of  {finding the}  $f$-domination number of a graph is called  the {\it $f$-domination problem}.   

 %As a motivation of this generalization, we give an application of $f$-dominating set in the area of the computer networks.  A {\it backbone} is a part of computer network that interconnects various pieces of network that  provides a path for exchange information between subnetworks.  From a practical point of view, consider a computer network. A dominating set can be used as a backbone of the network. If part of a network may be out of reach, the dominating set is not useful. But if we use $k$-dominating set instead of dominating set, until $k$ vertices of the network are not out of reach, the $k$-dominating set can be used as a backbone. Now, consider $\alpha$-dominating set as a backbone of a computer network. Actually, until there are a percent of neighbors of each vertex in the network, the $\alpha$-dominating set can be used as a backbone yet. Clearly, $\alpha$-dominating set can be more useful than $k$ -dominating set. Now, assume that we use $f$-dominating set as a backbone of a network. Actually, since in $f$-dominating set the number of neighbors of each vertex inside $f$-dominating set depends on the cardinality of the network, the $f$-dominating set can be more useful than $\alpha$-dominating set. 

{As a motivation of this generalization, we give an application of $f$-dominating set in the
area of the computer networks. A backbone is a part of a computer network that interconnects
various pieces of network that provides a path for exchange of information between subnetworks.
From a practical point of view, designing a backbone which is fault-tolerant is a major problem in designing computer networks.
The fault-tolerance means the networks retain its role, even if some parts (vertices or edges) fail. A dominating set can be used
as a backbone of the network, but a small fault can destroy the dominance of the backbone. If one uses a $k$-dominating set as the backbone, it remains as a dominating set until $k$ neighbors of a vertex in the backbone fails. The value of $k$ is constant and  it does not depend on vertex degrees.  Now, consider an $\alpha$-dominating set as a backbone of a computer network. Actually, the backbone retains its job as far as some portion of neighbors of each vertex in the backbone works. Clearly, using an $\alpha$-dominating set as the backbone is more useful than using $k$-dominating set. Now, assume that
we use $f$-dominating set as a backbone of a network. Actually, since in $f$-dominating set
the number of neighbors of each vertex inside $f$-dominating set is some depends on the value of function $f$ which is more general than the $\alpha$-dominating set,  this gives more flexibility and in some applications, $f$-dominating set for some special functions $f$ may work much better. Since the $f$-dominating set is a generalization of previous definitions of dominating sets (except for vector dominating set; see \cite{cicalese2013approximability}), in any application of dominating sets, one can use $f$-dominating set, instead. }

 %  {In applications of the dominating sets, one might be interested in finding a dominating set $D$ such that instead of  every vertex $v$ outside $D$ is dominated by an  $\alpha d_v$ vertices in $D$, as a linear function of $d_v$,  it is   dominated by a subset of vertices whose cardinality recognizes by a more complicated function of the degree of the vertex such as $\sqrt{d_v}+1$ and $2\ln(1+\frac{d_v}{2})$.}
 {In this paper, we prove that  for a large family of functions $f$, the following problem is NP-hard: given a graph $G$ and a positive integer $k$, decide whether $G$ has an $f$-dominating set $S$ with $|S|\leq k$.}
%%%%%%%%%%%%%%%%%%%%%
%%%%%%%%%%%%%%%%%%%%%%
\section{Approximation hardness and approximability}
In this section, we  discuss   approximation hardness of $k$-domination   and  $\alpha$-domination problems  in the case  {when the input instances are} restricted to {some} subclasses of graphs. For brevity, we  {denote} the $k$-domination problem by {\sc Min} $k$-{\sc Dom Set} and  $\alpha$-domination problem by {\sc Min} $\alpha$-{\sc Dom Set}.{ Moreover, in this section we  propose an  approximation algorithm for approximating of} {\sc Min} $k$-{\sc Dom Set} and {\sc Min} $\alpha$-{\sc Dom Set} for $p$-claw free graphs. % {We encourage the interested reader to consider the problem of $f_{deg}$-domination for other functions $f$ such as $f_{deg}(x)=\sqrt{x}$, $f_{deg}(x)=\ln x$ and etc.}   

Recently,   {the} problems {\sc Min} $k$-{\sc Dom Set},  {\sc Min} $\alpha$-{\sc Dom Set} and {\sc Min Dom Set} (special case of {\sc Min} $k$-{\sc Dom Set}, in which $k=1$)   {were considered} by some researchers \cite{chlebik2008approximation,cicalese2013approximability}.   {In the following, {we summarize the most relevant results, on which our new results will be based}}.
\begin{theorem}[\cite{cicalese2013approximability}]
\label{thm1kapx}
{\sc Min} $k$-{\sc Dom Set }  and {\sc Min} $\alpha$-{\sc Dom Set} can be approximated in polynomial time by a factor of $\ln(2\Delta(G))+1$, where $\Delta(G)$ is the maximum degree of $G$.
\label{kdapp}
\end{theorem}
Given a collection $\cal F$ of subsets of $S =\{1,\ldots,n\}$, {\sc Set Cover} is the problem of selecting as few as possible subsets
from $\cal F$ such that their union covers $S$ (see \cite{Feige1998}).  In 2014, Dinur et al. \cite{Dinur2014} showed  the following result. 
\begin{theorem}[\cite{Dinur2014}]
\label{thmdinur}
For any constant $\epsilon>0$ there is no polynomial time algorithm approximating {\sc Set Cover} within a factor of $(1-\epsilon)\ln n$ unless $P=NP$, where $n$ is the cardinality of the ground set.
\end{theorem}
Now, by applying  Theorem \ref{thmdinur}  to the results of \cite{chlebik2008approximation} we have:
\begin{theorem}
\label{thmepdom}
For any constant $\epsilon>0$ there is no polynomial time algorithm approximating {\sc Min Dom Set} within a factor of $(1-\epsilon)\ln n$ unless $P=NP.$ The same result {holds} for   {bipartite graphs and for split
graphs (hence also for chordal graphs).}
\end{theorem}
Also, by applying Theorem \ref{thmdinur} to the results of \cite{cicalese2013approximability} we have:
\begin{theorem}
\label{kdomeps}
For every $k \geq 1$ and every $\epsilon >0$, there is no polynomial time algorithm approximating {\sc Min} $k$-{\sc Dom Set} within a factor of $(1-\epsilon) \ln n$, unless $P=NP.$
\end{theorem}
\begin{theorem}
\label{thmalpha}
For every $\alpha\in(0,1)$ and every $\epsilon >0$, there is no polynomial time algorithm approximating {\sc Min} $\alpha$-{\sc Dom Set} within a factor of $(\frac{1}{2}-\epsilon) \ln n$, unless $P=NP.$
\end{theorem}
\subsection{{APX-completeness of} {\sc Min} $k$-{\sc Dom Set} {in graphs of degree at most $k + 2$}}
%\begin{theorem}
%For any fixed integer $k\geq 1$,  it is NP-hard to %approximate  {\sc Min} $k$-{\sc Dom Set}-$(k+2)$B %within $1+\frac{1}{?}$
%\end{theorem}
%\begin{proof}
%The proof is almost similar to proof of Theorem  \ref{kdomeps} in somewheres. Suppose $c=$.  We will make a reduction from  bipartite graphs $G$ of maximum degree 3 with $n$ vertices such that 
%\begin{flalign*}
%&n+2k-2\leq n^{1+c}. &
%\end{flalign*}
%and
%\begin{flalign*}
%&\gamma(G)\geq\frac{2(k-1)(1+c)}{c^2}.&
%\end{flalign*}
%It is clear that we can have the above assumption %without loss of generality. Also it is easy to see that 
%\begin{flalign*}
%&\gamma(G)+2(k-1)\leq\frac{1+c+c^2}{1+c}\cdot \gamma(G).&
%\end{flalign*}
%Let $G=(V_1, V_2, E)$ be a bipartite graph with above assumptions. We transform $G$ into a  graph $G'$ by adding to it two sets $K_1$ and $K_2$  such that each of sets $K_1$ and $K_2$ have $k-1$ vertices. The sets $K_1$, $K_2$ and $V(G)$ are pairwise disjoint. Join each vertex of $V_1$ to each vertex of $K_2$ and join each vertex of $V_2$ to each vertex of $K_1$. We must be note the induced graph on the sets $K_1$ and $K_2$ has no edges. We can easily prove $\gamma_k(G')\leq \gamma(G)+2k-2$ (see\cite{ff}).
%\end{proof}

%%%%%%%%%%%%%%%%%%%%%%

Now, we  use  Theorem \ref{kdapp} to prove that the problem {\sc Min} $k$-{\sc Dom Set} is APX-complete on  the graphs of   degree bounded by $k+2$  {for all constants $k\geq 2$}. For brevity,  we call the restricted problem {\sc Min} $k$-{\sc Dom Set} to the graphs  {of degree bounded by  $R$}  by {\textsc{Min}} $k$-{\textsc {Dom Set}}-$R$.   {Also in the case of}  {\textsc{Min}} $k$-{\textsc {Dom Set}}-$R$, when $k=1$, for brevity, we call the related problem by  {\textsc{Min}} {\textsc {Dom Set}}-$R$.

{First}, we recall  {the {notion} of $L$-reduction. }
\begin{definition}
\label{defdef}
($L$-reduction)\cite{ausiello}. Given two NP optimization problems $F$ and $G$ and a polynomial transformation  {$f:Inst(F)\rightarrow Inst(G)$, where $Inst(F)$ is the set of instances of $F$}; we say
that $f$ is an $L$-reduction if there are  {two} positive constants
$\alpha$ and $\beta$ such that for every instance $x$ of $F$
\begin{enumerate}
\item $opt_G(f(x))\leq \alpha opt_F (x)$
\item for every feasible solution $y$ of $f(x)$ with objective value $m_G(f (x),y) = c_2$, we can (in polynomial time) find a solution $y'$ of $x$ with $m_F (f (x),y')=c_1$ such that $|opt_F (x)-c_1|\leq \beta|opt_G(f (x))-c_2|$.
\end{enumerate}
\end{definition}
{To prove  a problem $F$ is APX-complete, it is {sufficient} to prove that}  $F\in$APX  {and that there} is an $L$-reduction from an APX-complete problem to the problem $F$.

\begin{theorem}
\label{thmApxkdom}
{\textsc{ Min}} $k$-{\textsc {Dom Set}}-$(k+2)$ is  {an} APX-complete  {problem} for any $k\geq 1$.
\end{theorem}
\begin{proof}
{It is known} that {\sc Min Dom Set}-3 is  APX-complete \cite{alimonti1997hardness}.  {So we consider the case $k>1$}. Clearly, by Theorem \ref{thm1kapx}  if {the vertex degrees of the graph are} bounded by a constant, then the approximation ratio is constant. Thus  {the} problem {\sc Min} $k$-{\sc Dom Set}-$(k+2)$ is in APX. Now, we show that there is  an $L$-reduction $f_k$ from {\sc Min Dom Set}-3 to {\sc Min} $k$-{\sc Dom Set}-$(k+2)$. Suppose that $G=(V, E)$ is a  graph  of   {maximum degree at most 3}. We construct a graph $G_k = (V_k, E_k)$ of   {maximum degree at most $k +2$} as follows.
For each vertex $v\in V$, suppose that $S_v$ is a set of $k-1$ new vertices, in particular,   we assume that the sets $S_v$ are pairwise disjoint and also for each $v\in V$, $S_v\cap V=\emptyset$. Now join each vertex $v \in V$ to each vertex of the set $S_v$.
 It is clear that the maximum degree of $G_k$  {is at most $k +2$.}

Now, we define the transformation $f_k$  {as} $f_k(G)=G_k$. Let $D_k$ be a $k$-dominating set in  $G_k$. We define the set $D$ as follows:
$$D=D_k-\left(\bigcup_{v\in V(G)}{S_v}\right).$$
Suppose that $v$ is a vertex of $G$ outside $D$. Clearly $v$ is outside $D_k$ too.  Since in the graph $G_k$ the vertex $v$  is adjacent to the elements of $S_v$  (recall $|S_v|=k-1$) and $v$  is dominated by at least $k$ vertices in $D_k$, clearly there must be  at least one vertex in  $D$ which dominates $v$. Hence, $D$ is a dominating set in $G$. Now, we have
\begin{equation}
\label{eqgam1}
\gamma(G)\leq |D|= |D_k|-(k-1)n,
\end{equation} 
where $n =|V|$.

On the other hand, let $D$ be  any dominating set in $G$. Clearly  the set
$$D_k=\left(\bigcup_{v\in V(G)}{S_v}\right)\cup D$$
is a $k$-dominating set in $G_k$. So 
\begin{equation}
\label{eqgamk}
\gamma_k(G_k)\leq |D_k|= |D|+(k-1)n.
\end{equation} 

Now suppose that $D_k^*$ and $D^*$ are optimal $k$-dominating set in $G_k$ and optimal dominating set in $G$, respectively.  According to  equations (\ref{eqgam1}) and (\ref{eqgamk}), we have   {$\gamma_k(G_k)=\gamma(G)+(k-1)n$, which} results in
$|D|-|D^*|=|D_k|-|D_k^*|.$

Also since the maximum degree of $G$ is bounded by 3,   any dominating set $D$ in $G$ contains at least  $\frac{n}{4}$ elements (see \cite{haynes1998fundamentals1}). Hence, since  $\gamma_k(G_k)=\gamma(G)+(k-1)n$, it is easy to see that  $|D^*_k|=|D^*|+(k-1)n\leq (4k-3)|D^*|$. Hence,
$$|D_k^*|\leq (4k-3)|D^*|,$$
 {and}  $f_k$ is an $L$-reduction with parameters  $\alpha=4k-3$ and $\beta=1$. So problem {\textsc{Min}} $k$-{\textsc {Dom Set}}-$(k+2)$ is APX-complete. 
\end{proof}
{Consider the  problem} {\textsc{Min}} $k$-{\textsc {Dom Set}}-R     restricted to bipartite graphs, denoted by  {\textsc{Min}} $k$-{\textsc {Dom Set}}-RB. Also, consider the  {problem} {\textsc{Min}} $k$-{\textsc {Dom Set}}-R     restricted to chordal graphs, denoted by  {\textsc{Min}} $k$-{\textsc {Dom Set}}-RC. 
For $k=1$, we call the  special  {cases} {\textsc{Min}} $k$-{\textsc {Dom Set}}-RB and {\textsc{Min}} $k$-{\textsc {Dom Set}}-RC by {\textsc{Min}} {\textsc {Dom Set}}-RB and {\textsc{Min}} {\textsc {Dom Set}}-RC, respectively. Chleb\'{\i}k and Chleb\'{\i}kov\'{a} \cite{chlebik2008approximation} proved that it is NP-hard to approximate  the problem {\sc Min Dom Set}-3B within factor $1+\frac{1}{390}$. Thus, we can easily conclude that the problem {\sc Min Dom Set}-3B is APX-complete.  Using  {the} same construction discussed in the proof of Theorem  \ref{thmApxkdom}, we can show that the problem {\sc Min} $k$-{\sc Dom Set}-$(k+2)$B is APX-complete. Hence we have the following result. 
\begin{theorem}
For any $k\geq  1$,   {the} problem {\sc Min} $k$-{\sc Dom Set}-$(k+2)$B is APX-complete.
\end{theorem}
{Note that for every positive integer $k$, the} {\sc {Min Dom Set}}-$k$C {problem is solvable in polynomial time.} This is because a chordal graph of bounded maximum degree has  bounded clique number and therefore it has  bounded treewidth (see \cite{BODLAENDER19981}). For any fixed $k$, the property that a subset of vertices is a $k$-dominating set can be straightforwardly expressed by a formula in Monadic Second
Order Logic. Therefore, the $k$-domination problem is polynomial time  solvable in any class of graphs of {bounded treewidth} \cite{Arnborg1991},  {and even more generally, in any class of graphs of bounded clique-width, using the
general results from} \cite{Courcelle2000,Oum2006}.  {Polynomial} time solvability of the $k$-domination problem in graphs of bounded clique-width follows also from the results in \cite{Cicalese2014},  where a dynamic programming solution was
developed for a more general problem, which includes the variant of $k$-domination where  $k$ is given as input.

\subsection{Inapproximability of {\sc Min} $k$-{\sc Dom Set} in bipartite graphs and split graphs}
In the following, we  present some results on the inapproximability of  {\sc Min} $k$-{\sc Dom Set}  when the problem  {is} restricted to some special classes of graphs.

 {By} some modifications of  {the} proof of Theorem \ref{kdomeps}, we prove that the result in Theorem~\ref{kdomeps} holds even {if} the  {problem is restricted} to  {bipartite graphs.} 
\begin{theorem}
\label{thm8}
For every $k \geq 1$ and every $\epsilon >0$, there is no polynomial time algorithms approximating {\sc Min} $k$-{\sc Dom Set} for bipartite graphs within a factor of $(1-\epsilon) \ln n$, unless $P=NP.$
\end{theorem}
\begin{proof}
It is {sufficient} that  we make the following modifications in the proof of Theorem \ref{kdomeps}.

We make a reduction from {the} domination {problem} on bipartite graphs. {Let  $G=(V_1, V_2, E)$ be a bipartite graph} with $n$ vertices such that $n+2k-2 \leq n^{1+\epsilon}$
and
 {$\gamma(G)\geq \frac{2(k-1)(1+\epsilon)}{\epsilon^2}$. }
Note that if  {$\gamma(G)< \frac{2(k-1)(1+\epsilon)}{\epsilon^2}$}, then with a brute-force algorithm  in polynomial time we can solve the problem {\sc Min} $k$-{\sc Dom Set}. Hence, without loss of generality,  we can have the above assumptions. Since  {$\gamma(G)\geq \frac{2(k-1)(1+\epsilon)}{\epsilon^2}$}, with some algebraic computations, we have 
\begin{equation}
\label{eqjaj3}
\gamma(G)+2k-2\leq\frac{1+\epsilon+\epsilon^2}{1+\epsilon} \times\gamma(G).\end{equation}
Now,  {transform} $G$ into a  graph $G'$ by adding {two new pairwise disjoint sets $K_1$ and $K_2$ } to $G$ such that each of sets $K_2$ and $K_1$ have $k-1$ vertices inducing a
 graph  {with} no edges.  Join each vertex of $V_1$ to each vertex of $K_1$ and join each vertex of $V_2$ to each vertex of $K_2$. Note that there is no edges between sets $K_1$ and $K_2$. Obviously graph $G'$ is bipartite.  It is not hard to see that if $D$ is a dominating set in $G$, then $D\cup K_1\cup K_2$ is a $k$-dominating set in $G'$. Thus  $\gamma_k(G')\leq \gamma(G)+2k-2.$

Now  suppose that there is a polynomial time approximation algorithm that {computes} a $k$-dominating set $D'$ in $G'$ such that  $|D'|\leq (1-\epsilon)\ln(|V(G')|)\gamma_k(G')$. It is not hard  to see that  $D:=D'\cap V(G)$ is a dominating set in  $G$. So we have  
\begin{align}
\nonumber
|D|&\leq |D'|\\\nonumber
&\leq  (1-\epsilon)(\ln |V(G')|) \gamma_{k}(G')\\\nonumber
&\leq(1-\epsilon)(\ln (n+2k-2))(\gamma(G)+2k-2)\\ \nonumber
&\leq(1-\epsilon)(\ln n^{1+\epsilon})  (\gamma(G)+2k-2) \mbox{~(since $n+2k-2 \leq n^{1+\epsilon}$)  } \\\nonumber 
&\leq(1-\epsilon)(\ln n^{1+\epsilon})  \left(\frac{1+\epsilon+\epsilon^2}{1+\epsilon}\right)\gamma(G) \mbox{ (by Equation \ref{eqjaj3}) }\\\nonumber
&=(1-\epsilon)(\ln n)  ({1+\epsilon+\epsilon^2})\gamma(G)\\\nonumber
&=(1-\epsilon')(\ln n)\gamma(G),
\end{align}
where $\epsilon'=\epsilon^3>0.$  So  {$D$ approximates the domination number in  $G$}  within factor $(1-\epsilon')\ln n$.  {By Theorem} \ref{thmepdom}, {this implies that P=NP}. Hence there is no polynomial time algorithms approximating {\sc Min} $k$-{\sc Dom Set} for bipartite graphs within a factor of $(1-\epsilon) \ln n$, unless P=NP.
\end{proof}
A {\it split graph}  {is a graph {whose} vertices can be partitioned into an independent set and a clique, and a} {\it chordal graph}  {is a graph such that   every cycle with  four or more vertices has a chord}. It is clear that in the proof of Theorem \ref {kdomeps} if $G$ is  {a split} (chordal) graph, then $G'$ is a split (chordal) graph too. Hence  the following result holds. 
\begin{corollary}
\label{col1}
For every $k \geq 1$, every $\epsilon >0$, there is no polynomial time algorithm approximating {\sc Min} $k$-{\sc Dom Set} for  {split (chordal)} graphs within a factor of $(1-\epsilon) \ln n$, unless $P=NP$.
\end{corollary} 
\subsection{Inapproximability of {\sc Min} $\alpha$-{\sc Dom Set} in bipartite graphs and chordal  graphs}
Here we  present some results on inapproximability of {\sc Min} $\alpha$-{\sc Dom Set} when the problem  {is} restricted to bipartite graphs and chordal graphs. 

{In the proof of Theorem} \ref{thmalpha}, Cicalese  et al.     first showed that for every integer $B > 0$ and for every $\epsilon > 0$, there is no polynomial time algorithm approximating  {the} domination  {number of} an input graph $G$ without isolated vertices  {and} satisfying $\gamma (G) \geq B\Delta(G)$ within a factor of $(\frac{1}{2}-\epsilon)\ln n$, unless NP$\subseteq$ DTIME$(n^{O(\log\log n)})$.  {Looking} at the proof of this result and using Theorem \ref{thmepdom}, it is easy to see that the result holds for bipartite graphs and {chordal graphs, and under the weaker assumption P$\neq$ NP.} Hence we have:
\begin{lemma}
\label{lemlem}
For every integer $B > 0$ and for every $\epsilon > 0$, there is no polynomial time algorithm approximating domination on input bipartite graphs $G$ without isolated vertices satisfying $\gamma (G) \geq B\Delta(G)$ within a factor of $(\frac{1}{2}-\epsilon)\ln n$, unless $P=NP$. {The} same result holds for chordal graphs.
\end{lemma}
With some modifications in   {the} proof of Theorem \ref{thmalpha} and using Lemma \ref{lemlem} we  obtain the following  {result. }
\begin{theorem} 
\label{thm9}
For every $\alpha\in(0,1)$ and  $\epsilon >0$, there is no polynomial time algorithms  approximating {\sc Min} $\alpha$-{\sc Dom Set} for bipartite graphs within a factor of $(\frac{1}{2}-\epsilon) \ln n$, unless $P=NP$.
\end{theorem}
\begin{proof}
Let $0<\alpha<1$ and $\epsilon \in (0, 1)$.    We define $N = \left\lceil\frac{\alpha}{1-\alpha}\right\rceil$,   $B=\left\lceil\frac{2N}{\epsilon}\right\rceil$ and     {$k=N\Delta(G)$}. The reduction is  {done from domination in bipartite}
graphs $G$ without isolated vertices and with $n$ vertices such that $1+2N\leq n^{\epsilon}$ and $\gamma(G)\geq B\Delta(G).$

Using  {the} same transformation as in  {the} proof of Theorem \ref{thmalpha}, we transform bipartite graph $G=(V_1, V_2, E)$ into a graph $G'$ as follows: {consider} two sets $K_1$ and $K_2$  each  {with} $N\Delta(G)$ extra vertices. We assume that sets $K_1$, $K_2$ and $V(G)$ are pairwise disjoint. Join each vertex $v\in V_1$ to precisely $k_v$ vertices of $K_2$  and join each vertex $v\in V_2$ to precisely $k_v$ vertices of $K_1$  {where}, 
\begin{align}
\nonumber
&k_v=\left\{\begin{array}{cc}
\left\lceil\frac{\alpha d_G(v)-1}{1-\alpha}\right\rceil& \text{if~} d_G(v)\geq2\\
0& \text{if~} d_G(v)=1.
\end{array}\right.&
\end{align}

Clearly  $0\leq k_v\leq N\Delta(G)$ when $d_G(v)\geq 2$ and $k_v=0<N\Delta(G)$ when $d_G(v)=1$.  {Thus,} the above transformation can be done. 

It is easy to prove that $\gamma(G)\leq\gamma_{\alpha}(G')\leq \gamma(G)+2k$ (see \cite{cicalese2013approximability}).  {Suppose that  there exists} a polynomial time algorithm $A$ that computes an $\alpha$-dominating set $S'$ for $G'$ such that $|S'|\leq (\frac{1}{2}-\epsilon)\ln (|V(G')|)\gamma_{\alpha}(G')${.} It is clear that $|V(G')|=n+2k=n+2N\Delta(G)\leq n(1+2N)\leq n^{1+\epsilon}$.

Using the inequality $\gamma(G)\geq\Delta(G)B$, it is easy to see that $2k\leq \epsilon\gamma(G)$. Furthermore, it is straightforward that   $S = S' \cap V(G)$ is a dominating set in $G$ (see \cite{cicalese2013approximability}). Hence, we have 
\begin{align}
\nonumber
|S|&\leq |S'|&\\\nonumber
&\leq  (\frac{1}{2}-\epsilon)(\ln |V(G')|) \gamma_{\alpha}(G')\\\nonumber
&\leq(\frac{1}{2}-\epsilon)(\ln n^{1+\epsilon})  (\gamma(G)+2k)\\\nonumber
&\leq(\frac{1}{2}-\epsilon)(\ln n^{1+\epsilon})  (\gamma(G)+\epsilon\gamma(G))\\\nonumber
&=(\frac{1}{2}-\epsilon)(1+\epsilon)^2(\ln n) \gamma(G)\\\nonumber
&=(\frac{1}{2}-\epsilon')(\ln n)\gamma(G),
\end{align}
where  $\epsilon':=\epsilon^2(\epsilon + 3/2) \in (0, \frac{1}{2}).$ Thus, using algorithm $A$, we find a set $S$ in polynomial time that approximates {\sc Min Dom Set} within a factor of $(\frac{1}{2}-\epsilon') \ln n$ {, which implies that $P=NP$}. Hence, there is no polynomial time algorithm approximating   {the problem} {\sc Min} $\alpha$-{\sc Dom Set}  for bipartite graphs  within a factor of $(\frac{1}{2} -\epsilon) \ln n$, unless   $P=NP$.
\end{proof}
Using  {the} same construction of  {the} proof of Theorem \ref{thmalpha} and using Lemma \ref{lemlem} we can easily prove the following result:
 \begin{theorem} 
 \label{thm10}
For every $\alpha\in(0,1)$,  integer $B>0$ and  $\epsilon >0$, there is no polynomial time  {algorithm} approximating  {the} {\sc Min} $\alpha$-{\sc Dom Set}  {problem} for chordal graphs within a factor of $(\frac{1}{2}-\epsilon) \ln n$, unless $P=NP$.
\end{theorem}
\begin{proof}
The proof  is  {the} same as  {the} proof of Theorem \ref{thmalpha}. Note that  the transformation of chordal graph $G$ to $G'$ can be done such that $G'$ be a chordal graph. 
\end{proof}
%%%%%%%%%%%%%%%%%%%%%%%%%
%%%%%%%%%%%%%%%%%%%%%%%%%
%%%%%%%%%%%%%%%%%%%%%%%%%%%%%%%%%
%%%%%%%%%%%%%%%%%%%%%%%%%%%%%%%%%
\subsection{Approximation algorithm for  {\sc Min} $k$-{\sc Dom Set} and {\sc Min} $\alpha$-{\sc Dom Set} in $p$-claw free graphs}
Here we  give a simple polynomial time approximation algorithm that gets a $p$-claw free graph $G$  {as  input and  for a constant $k$,    computes} a $k$-dominating set in $G$ with approximation ratio $\max\{p-1,k\}$.  We also present similar results in the case of {\sc Min} $\alpha$-{\sc Dom Set}. 

{First, we prove a relation} between an optimal  {$k$-dominating set} and maximal independent set,  {denoted by MIS,}  of a $p$-claw free graph in analogy with Lemma 7 in \cite{klasing2004hardness}. 
\begin{lemma}
\label{lemlem2}
Let $D$ be any optimal $k$-dominating set  of $G$ and $I$ be  {any MIS} of $G$, where $G$ is a $p$-claw free graph and $k$ is a constant.  Then,  {$|{D}|\geq \min\left\{ \frac{k}{p-1},1\right\}\cdot|I|.$}
\end{lemma}
\begin{proof}
 For all $u\in I$ and for all $v\in {D}$, let 
$x_u=|{D}\cap N[u]|, y_v=|I\cap N[v]|, x'_u=|{D}\cap N(u)|$ and $y'_v=|I\cap N(v)|$,
where $N[v]=N(v)\cup\{v\}$. Since ${D}$ is  {a} $k$-dominating set in $G$,  $I$ is {a} maximal independent set  and $G$ is $p$-claw free,  we have:
\begin{itemize} 
\item If $u\in I-{D}$, then $x_u\geq k$ and if $u\in I\cap{D}$ then $x_u\geq 1$. Thus 
\begin{flalign} 
\label{b1}
&\sum_{u\in I}{x_u}\geq k|I-{D}|+|I\cap {D}|&
\end{flalign}
\item If $v\in {D}-I$ then $y_v\leq p-1$ and if $v\in I\cap{D}$, then $y_v=1$. Thus 
\begin{flalign} 
\label{b2}
&\sum_{v\in {D}}{y_v}\leq (p-1)|{D}-I|+|I\cap {D}|&
\end{flalign}
\item  If $u\in I-{D}$, then $x_u=x'_u$ and if $u\in I\cap{D}$, then $x_u=x'_u+1$. Thus 
\begin{flalign} 
\label{b3}
&\sum_{u\in I}{x_u}=\sum_{u\in I}{x'_u}+|I\cap {D}|&
\end{flalign}
\item  If $v\in {D}-I$ then $y_v=y'_v$ and if $v\in I\cap{D}$, then $y_v=y'_v+1$. Thus 
\begin{flalign} 
\label{b4}
&\sum_{v\in {D}}{y_v}=\sum_{v\in {D}}{y'_v}+|I\cap {D}|&
\end{flalign}
\end{itemize}
By considering the edges between sets $I$ and ${D}$, it is easy to see that $\sum_{v\in {D}}{y'_v}=\sum_{u\in I}{x'_u}$. Consequently,  $\sum_{v\in {D}}{y_v}=\sum_{u\in I}{x_u}$.  Hence,  according to   {equations} (\ref{b1}) and  (\ref{b2}), we have
$$(p-1)|{D}-I|+|{D}\cap I|\geq k|I-{D}|+|{D}\cap I|.$$
With replacing $|{D}-I|$ by $|{D}|-|{D}\cap I|$ and $|I-{D}|$ by $|I|-|I\cap{D}|$ in  above formula, we have
$$|{D}|\geq\frac{k}{p-1}|I|+\frac{p-k-1}{p-1}|{D}\cap I|.$$
Now if $p\geq k+1$, then $|{D}|\geq \frac{k}{p-1}|I|$ and if $p< k+1$, since $|{D}\cap I|\leq |I|$,  then
$$|{D}|\geq\frac{k}{p-1}|I|+\frac{p-k-1}{p-1}|{D}\cap I|\geq\frac{k}{p-1}|I|+\frac{p-k-1}{p-1}|I|=|I|.$$
Hence, $|{D}|\geq \min\left\{ \frac{k}{p-1},1\right\}\cdot|I|.$
\end{proof}
Now we  present an approximation algorithm. 

\begin{algorithm}[H]
\SetKwInOut{Input}{Input}
\Input{A $p$-claw free graph $G$.} 
\SetKwInOut{Output}{Output}
% \SetAlgoLined
\Output{A $k$-dominating set $D$ in $G$.}
\For {$i:=1$ to $k$}{Construct  {an} MIS $I_i$ in $G-I_1 \cup\ldots\cup I_{i-1}$\;
}
$D := I_1 \cup\ldots\cup I_{k}$\;

\caption{\sc{ {Approximate-$k$-Dom-Claw$(G)$}  }}
 \end{algorithm}
\begin{theorem}
\label{thmalg1}
{The} algorithm {\sc Approximate-$k$-Dom-Claw$(G)$} computes  {in polynomial time}  a $k$-dominating set $D$ in {a} $p$-claw free graph $G$ such that  $|D|\leq \max\left\{p-1,k\right\}\cdot\gamma_{k}(G)$. 
\end{theorem}
\begin{proof}
It is well known  that a maximal independent set can be computed in polynomial time in any graph by a simple greedy algorithm. Therefore, the algorithm {computes}  the set $D$ in polynomial time. 

Now we prove that $D$ is a $k$-dominating set in $G$. Suppose that $v\not \in D $. So $v\not \in I_i$, for $1\leq i\leq k$. Since $I_i$ are   maximal independent {sets}, vertex $v$ has at least one neighbor in each set $I_i$ for each $1\leq i\leq k$. Thus  $D$ is a $k$-dominating set in $G$. 

 According to the construction of  $D$ and by applying Lemma \ref{lemlem2}, 
$$|D| =\sum_{i=1}^{k}|I_i|\leq \sum_{i=1}^{k} \max\left\{\frac{p-1}{k},1\right\}\gamma_k(G)=\max\left\{p-1,k\right\}\cdot\gamma_{k}(G).$$
\end{proof}
 {Let $k$ and $p$ be two positive integers with $p \geq k + 1$ and let $G$ be a $p$-claw free graph.} If $\Delta_G>\frac{e^{p-2}}{2}$,  where $e$ is  {the} Euler's number, then clearly $p-1<\ln (2\Delta_G)+1$. So in this case, the approximation ratio in algorithm  {\sc Approximate-$k$-Dom-Claw} is  {better} than the approximation { ratio given by}   Theorem \ref{thm1kapx}. Now suppose that $H$ is a graph with $\delta_H\geq 3$,  {where $\delta_H$ is the minimum degree of $H$},  and suppose that $G=L(H)$ is the line graph of graph $H$ ( {the line} graph $L(H)$ of graph $H$ is a graph such that  each vertex of $L(H)$ represents an edge of $H$;    {two vertices} of $L(H)$ are adjacent if and only if their corresponding edges share a common endpoint). It is easy to see that $G$ is a 3-claw free graph with $\Delta_G>\frac{e^{p-2}}{2}=\frac{e}{2}$.   {In the case $p \geq k + 1$, the algorithm is a 2-approximation for $k$-domination, where $k\in\{1,2\}$. In the case $k = 1$, the algorithm is equivalent to computing a maximal matching in $H$ (where $G = L(H)$), which is 
 a well-known 2-approximation algorithm for  the minimum edge domination problem in $H$. }

 % {In the  Figure }\ref{figpclaw},  {we presented a 3-claw free graph $G$ such that $\Delta_G=4>\frac{e^{p-2}}{2}=\frac{e}{2}$.}
%\begin{figure}[!ht]\centering\includegraphics[width=6cm]{pclaw.pdf}\caption{ {3-claw free graph $G$ suchthat}$\Delta_G=4>\frac{e^{p-2}}{2}=\frac{e}2}$.}\label{figpclaw}\end{figure}

 {In a similar} way, suppose that  $p< k+1$. If 
$\Delta_G>\frac{e^{k-1}}{2}$ then clearly $k<\ln (2\Delta_G)+1$. So in this case, the approximation  {ratio given by} algorithm  {\sc Approximate-$k$-Dom-Claw} is  {better}   than the approximation  {ratio given by} Theorem \ref{thm1kapx}. In this case, one can find  many $p$-claw free graphs with $\Delta_G>\frac{e^{k-1}}{2}$.

In the following, we show that 
Algorithm {\sc Approximate-$k$-Dom-Claw$(G)$} for $k:=\lceil\alpha\delta_G\rceil$  computes an $\alpha$-dominating set $D$ for a given $p$-claw free graph $G$ with approximation ratio $\max\{p-1,k\}$.
\begin{lemma}
\label{lemlem3}
Let ${D}$ be any optimal {$\alpha$-dominating} set  in $G$ and let $I$ be any maximal independent set  of $G$, where $G$ is a $p$-claw free graph and $\alpha$ is a constant such that $0<\alpha<1$.  Then
$|{D}|\geq \min\left\{ \frac{\lceil\alpha\delta_G\rceil}{p-1},1\right\}\cdot|I|.$
\end{lemma}
\begin{proof}
Suppose that $k'=\lceil\alpha\delta_G\rceil$. Clearly $|{D}|\geq \gamma_{k'}(G)$. So by Lemma \ref{lemlem2}, we have 
$$|{D}|\geq \gamma_{k'}(G)\geq\min\left\{ \frac{\lceil\alpha\delta_G\rceil}{p-1},1\right\}\cdot|I|.
$$
\end{proof}
\begin{theorem}
Algorithm {\sc Approximate-$k$-Dom-Claw$(G)$} for $k:=\lceil\alpha\delta_G\rceil$ computes in polynomial time {an} $\alpha$-dominating set $D$ in  {a} $p$-claw free graph $G$ such that $|D|\leq \max\{p-1,k\}\cdot\gamma_{\alpha}(G)$.
\end{theorem}
\begin{proof}
Proof is similar to  {the} proof of Theorem \ref{thmalg1},  {just use Lemma} \ref{lemlem3}  {instead of}  Lemma~\ref{lemlem2}.
\end{proof}
%Suppose that $G$ is a $p$-claw free graph and $\alpha$ is a constant such that $0<\alpha<1$. If $\Delta_G>\frac{1}{2}e^{k \max\left\{ \frac{p-1}{ \alpha \delta_G},1\right\}-1}$,  then clearly $ k \max\left\{ \frac{p-1}{\alpha \delta_G},1\right\}<\ln (2\Delta_G)+1$. So, in this case,  the approximation ratio $k\cdot \max\left\{ \frac{p-1}{\alpha \delta_G},1\right\}$ is better than the approximation  {ratio given by} Theorem \ref{thm1kapx}.

\section{NP-completeness  result}
{We recall the definition of {an} $f$-dominating set {and of the $f$-domination number}. Given a function  $f:\mathbb{N}\rightarrow \mathbb{R}$, where $\mathbb{N}=\{1, 2, 3, \ldots\}$,  a set $D\subseteq V$ is called an  $f$-dominating set in $G$ if for every vertex $v$ of $G$ outside $D$, $|N(v)\cap D|\geq f(d_v)$. The cardinality of a minimum $f$-dominating set in a graph $G$ is called the $f$-domination number of $G$, and the problem of  finding the  $f$-domination number of a graph is called  the $f$-domination problem. }

In this section,  we   prove that the problem of finding {the} $f$-domination number of a graph is NP-complete, for every given function $f$ with some special properties.  It is well known that the following decision problem, denoted by 3-Regular Domination(3RDM), is NP-complete \cite{garey2002computers}: given a 3-regular graph $G =(V, E)$ {and} a positive integer $k$, does  $G$ have a dominating set $S$ of size at most  $k$?  Now,  consider the following decision problem, denoted by  $f$-Domination ($f$DM): given a graph $G =(V, E)$ without isolated vertices {and a} positive integer $k$, does $G$ have an $f$-dominating set $S$ of size at most  $k$?

 We  show that $f$DM is NP-complete for some special functions. We  extend  the proof of NP-completeness of  $\alpha$-domination problem (see \cite{dunbar2000alpha}). 
\begin{theorem}
Let $f : \mathbb{N} \rightarrow \mathbb{R}$ be a polynomially computable function such that $\exists x,y\in \mathbb{N}$ satisfying $x=\lceil f(y)\rceil<y$ and $x+1= \lceil f(x + 3)\rceil$. Then, the $f$DM problem is NP-complete.
 %{Let $f$ be a (strictly) increasing and   {polynomially computable \mbox{real-valued function}} with domain $\mathbb{N}$ that satisfies}
%\begin{enumerate}[label=\bf \alph*.]
%\item $\forall x\in \mathbb{N}, f(x) > 0.$
%\item $\exists x, y\in \mathbb{N}$ such that $f(x+3)\leq x + 1 < f(x + 2) + 1, f(y)\leq x \leq y$  {and} \mbox{$x < f(y + 1)\leq f(y)+1$.}%\exists x_0>0 \mbox{ such that }\forall x\geq x_0, ~ f_{deg}(x+3)\leq x+1.$
%end{enumerate}
%then, {the  $f$DM problem is  NP-complete.}
\label{thmNPf}
\end{theorem}
\begin{proof}
Since $f$ is a polynomially computable function, the membership in NP is trivial. Now, we prove the NP-hardness.

We make a transformation from 3RDM to $f$DM. We assume that $f$ is a fixed function that satisfies {the conditions} of the theorem. {Consider  a pair $x,y$ satisfying  the conditions  of the theorem}. 
% Suppose that   $x$ is the smallest positive integer satisfying the above condition the smallest  integer such that $(x+1)\geq f_{deg}(x +3)$. Now suppose that $y_x$ be the largest integer with $y_x>x$ and $x\geq f_{deg}(y_x)$.
 Also, suppose that $K_{y+1}$ is a complete graph on $y +1$ vertices. We denote the vertex set of $K_{y+1}$ by $W$ and let  $U$ be  {a subset of $W$ with $x$ elements}.

Now assume that  $G$ is  a 3-regular graph. We transform graph $G$ to a graph denoted by $\hat{G}$ as follows:  join each vertex of   $U$ to all vertices of $G$ (see Figure \ref{Gstar}). It is easy to see that since  $f, x$ and $y$ are fixed,  the transformation can
be done in polynomial time. Now we prove that $G$ has a dominating set $S$ of size of at most  $k$ if and only if $\hat{G}$ has an $f$-dominating set $D$ of size of at most $x +k$.

First, we assume that $S$ is a dominating set in $G$ such that  $|S|\leq k$. Consider  {the} set $D = S \cup U$. Using the  conditions of the theorem, it is easy to see that $D$ is an $f$-dominating set in  $\hat{G}$ with $|D|\leq x+k$.

Now assume that there  is an $f$-dominating set $D$ in  $\hat{G}$ with at most $x+k$ elements. %Between all $f_{deg}$-dominating sets in $\hat{G}$ with at most $x+k$ elements,  let $D$ be the one that has  maximum number of elements of  $U$. Also, without loss of generality, we can suppose that  there is  a vertex in $W-U$ that is outside  $D$.
 In the following, we  show that there is a dominating set $S$ in $G$ of size at most $k$. We consider two cases: $W-U\nsubseteq D$ and  $W-U\subseteq D$. 
 
First, suppose that $W-U\nsubseteq D$. We choose  $S:=D \cap V(G)$.  
We prove that  $S$ is  a dominating set in $G$ with $|S|\leq k$.
First, we prove that  $S$ is a dominating set in $G$.  Suppose that  $v\in V(G)-D$. {Since $D$ is an $f$-dominating set for $\hat{G}$ and the degree of vertex $v$ is $x +3$, $|N(v) \cap D|\geq f(x+3)$.  Hence, by the condition $x+1=\lceil f(x+3)\rceil$, we have  $|N(v) \cap D|\geq x + 1$, which implies that   $v$ has a neighbor in $S$.  As a result,  $S$ is a dominating set in $G$.}

Now, we prove that $|S|\leq k$.  Since $W-U\nsubseteq D$, then there is a vertex $u$ such that  $u\in W-U$ and $u\not\in D$. Since $D$ is an $f$-dominating set for $\hat{G}$, we have $|N(u)\cap D|\geq f(y)$. Hence, by the assumption $x=\lceil f(y)\rceil$, $|N(u)\cap D|\geq x$, and therefore   $|D \cap W|\geq x$, which means that  $|S|\leq k$. 

Now, suppose that $W-U\subseteq D$. We claim that using the set $D$, we can construct an $f$-dominating set $D'$ for $\hat{G}$ such that $|D'|\leq x+k$ and $W-U\nsubseteq D'$.  Hence, by the similar reasons for  the first case ($W-U\nsubseteq D$),  this will imply that $S:=D'\cap V(G)$ is a dominating set in $G$ with $|S|\leq k$.

It remains to prove the claim. We consider two cases: $W\subseteq D$ and $W\nsubseteq D$.  If $W\subseteq D$, then let $p$ be an arbitrary vertex in $W-U$. Note that by the  assumption $x<y$, there always exists such a vertex $p$. Now, suppose that $D':=D\backslash\{p\}$. It is obvious that $|D'|<|D|\leq x+k$. Now, since $p$ is not adjacent to any  vertex in  $G$ and also since $K_{y+1}$ is a complete graph, it is not hard to see that   $D'$ is an $f$-dominating set in $\hat{G}$.

Now, suppose that $W\nsubseteq D$. Since we assumed that $W-U\subseteq D$, there exists a vertex $q\in U-D$. Let $p$ be an arbitrary vertex in $W-U$. We choose  $D':=\left(D\backslash\{p\}\right)\cup \{q\}$. Obviously, $|D'|=|D|\leq x+k$.   Again, since $p$ is not connected to any  vertex in  $G$ and also since $K_{y+1}$ is a complete graph, clearly  $D'$ is an $f$-dominating set in $\hat{G}$.

Because 3RDM  is  NP-complete \cite{garey2002computers}, $f$DM is also NP-complete for the function $f$, which satisfies the conditions of Theorem \ref{thmNPf}.
\end{proof}
\begin{figure}[t]
\begin{center}
	\includegraphics[scale=1]{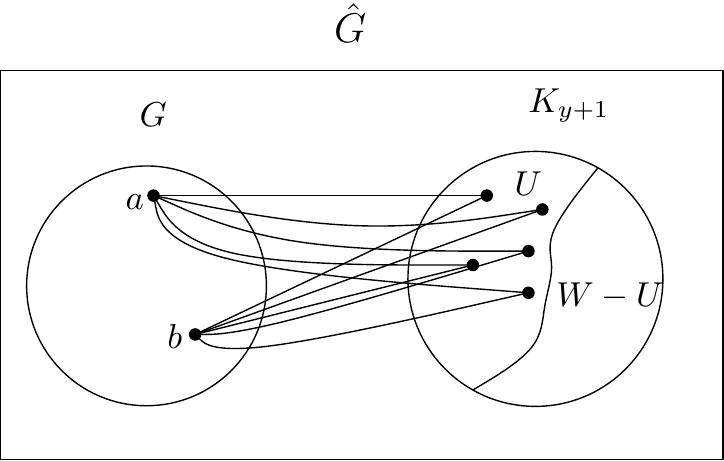}
	\caption{Construction of graph $\hat{G}$: each vertex of $U$ is joined to all  {vertices} of $G$.}
   \label{Gstar}
	\end{center}
\end{figure}
There are many functions that satisfy the conditions of Theorem \ref{thmNPf}, such as $\frac{x}{2}$, $\sqrt{x}+1$ and  $2\ln\left(1+ \frac{x}{2}\right)$. In Table \ref{tabel1}, we present concrete examples of  pairs $(x,y)$  satisfying the conditions of Theorem \ref{thmNPf}  for these three functions. 
\begin{table}
\begin{center}
\begin{tabular}{ |c|c| }
\hline
Function&$(x,y)$\\
  \hline			
  $\frac{x}{2}$ & $(1,2)$  \\[1ex]
  \hline
   $\sqrt{x}+1$& $(3,4)$ \\[1ex]
   \hline
  $2\ln\left(1+ \frac{x}{2}\right)$ & $(2,3)$  \\[1ex]
  \hline  
\end{tabular}
\end{center}
\caption{Pairs $(x,y)$ satisfying the conditions of Theorem \ref{thmNPf}.}
\label{tabel1}
\end{table}
%%%%%%%%%%%%%%%%%%%%%%%%%%%%
\section{Concluding remarks}
In this paper, we introduced the concept of $f$-domination as a  generalization of $\alpha$-domination. Furthermore,  we presented some approximability and inapproximability results on  the \mbox{problems} of finding  the $k$-domination number and the  $\alpha$-domination number for some  classes of graphs.  Furthermore, we proved the NP-completeness of $f$-domination problem {under mild assumptions on function $f$}.  It is remarkable that  the family of $f$DM problems can be seen as a generalization of the {\mbox{$\alpha$-domination}} problems and as a special case of the vector domination problem.  The vector domination problem is defined as follows.  Given a graph $G= (V, E)$ with  $n$ vertice and an $n$-dimensional  non-negative vector $\left(k_v; v\in V\right)$ such that for all $v\in V$, $k_v\in\{0,1,\ldots,d_v\}$; the  {\it vector domination}   is the problem of finding a minimum $S \subseteq V$ such that every vertex  $v\in V\backslash S$ has at least $k_v$ neighbors in $S$. We refer the reader for reading about the vector domination to \cite{cicalese2013approximability,ishi2016}. The connection between $f$-domination and vector domination problems implies that for every polynomially computable function $f$ and every class of graphs  in which the vector domination problem is polynomially solvable, the $f$DM problem is also polynomially solvable.

Finally, we leave open the following problem. 

Could we identify some  interesting classes of graphs where the $f$DM problem can be efficiently solved for many choices of function $f$?
\section*{Acknowledgments}
The authors would like to thank the reviewer for his/her very insightful comments that {improved} the paper, both in structure and in the results, especially for pointing out the reference \cite{Dinur2014} to the authors that improves the results in {Theorems} \ref{thmepdom}, \ref{kdomeps}, \ref{thmalpha}, \ref{thm8}, \ref{thm9}, \ref{thm10}, Corollary \ref{col1} and Lemma \ref{lemlem}.


\begin{thebibliography}{10}
\providecommand{\url}[1]{{#1}}
\providecommand{\urlprefix}{URL }
\expandafter\ifx\csname urlstyle\endcsname\relax
  \providecommand{\doi}[1]{DOI~\discretionary{}{}{}#1}\else
  \providecommand{\doi}{DOI~\discretionary{}{}{}\begingroup
  \urlstyle{rm}\Url}\fi

\bibitem{alimonti1997hardness}
{P. Alimonti  and V. Kann:}
\newblock Hardness of approximating problems on cubic graphs.
\newblock In G.~Bongiovanni, D.~Bovet, and G.~Di~Battista, editors, {Algorithms and Complexity}, LNCS, {\bf 1203}, 288--298 (1997).

\bibitem{Arnborg1991}
{S. Arnborg, J. Lagergren, and D. Seese:}
\newblock Easy problems for tree-decomposable graphs.
\newblock  {J. Algorithms}, {\bf12}, 308--340, (1991).

\bibitem{ausiello}
{ G. Ausiello:}
\newblock {Complexity and approximation: Combinatorial optimization
  problems and their approximability properties}.
\newblock Springer Science Business Media, (1999).

\bibitem{BODLAENDER19981}
{H. L. Bodlaender:}
\newblock A partial $k$-arboretum of graphs with bounded treewidth.
\newblock {Theoret. Comput. Sci.,} {\bf 209}, 1--45, (1998).

\bibitem{chlebik2008approximation}
{ M. Chleb{\'\i}k and J. Chleb{\'\i}kov{\'a}:}
\newblock Approximation hardness of dominating set problems in bounded degree
  graphs.
\newblock {Inform. Comput.}, {\bf 206}, 1264--1275, (2008).

\bibitem{Cicalese2014}
{ F. Cicalese, G. Cordasco, L. Gargano, M. Milani{\v{c}}, and U. Vaccaro:}
\newblock Latency-bounded target set selection in social networks.
\newblock {Theoret. Comput. Sci.,} {\bf 535}, 1--15, (2014).

\bibitem{cicalese2013approximability}
{ F. Cicalese,  M. Milani{\v{c}} and U. Vaccaro:}
\newblock On the approximability and exact algorithms for vector domination and
  related problems in graphs.
\newblock {Discrete. Appl. Math.,} {\bf 161}, 750--767, (2013).

\bibitem{Courcelle2000}
{ B. Courcelle, J. A. Makowsky, and U. Rotics.:}
\newblock Linear time solvable optimization problems on  graphs of bounded clique-width.
\newblock {Theory Comput. Syst.,} {\bf 33}, 125--150, (2000).
\bibitem{Dinur2014}
{I. Dinur and D. Steurer:}
\newblock Analytical Approach to Parallel Repetition.
\newblock In proceedings of the 46th Annual ACM Symposium on Theory of Computing. 624--633, (2014). 
\bibitem{dunbar2000alpha}
{  J. E. Dunbar,   D. G. Hoffman,  R. C. Laskar and L. R. Markus:}
\newblock $\alpha$-domination.
\newblock {Discrete. Math.}, {\bf 211},  11--26, (2000).

\bibitem{Feige1998}
{U. Feige:}
\newblock A Threshold of $\ln n$ for Approximating Set Cover.
\newblock {J. ACM}, {\bf 45}, 634--652, (1998).

\bibitem{fink1985n2}
{ J. F. Fink and M. S. Jacobson: }
\newblock$n$-domination in graphs.
\newblock In {Graph theory with applications to algorithms and computer
  science},  John Wiley \& Sons, Inc., 283--300,  (1985).

\bibitem{fink1985n1}
{ J. F. Fink and M. S. Jacobson:}
\newblock On $n$-domination, $n$-dependence and forbidden subgraphs.
\newblock In {Graph theory with applications to algorithms and computer
  science},  John Wiley \& Sons, Inc.,  301--311, (1985).

\bibitem{garey2002computers}
{M. R. Garey and   D. S. Johnson:}
\newblock {Computers and intractability: A Guide to the Theory of
  NP-Completenes}.
\newblock Freeman, New York, (1979).

\bibitem{haynes1998fundamentals1}
{ T. W. Haynes, S. Hedetniemi and P. J. Slater:}
\newblock {Fundamentals of domination in graphs}.
\newblock CRC Press, (1998).

\bibitem{haynes1998domination}
{  T. W. Haynes, S. Hedetniemi and P. J. Slater: }
\newblock {Domination in graphs: advanced topics}.
\newblock Marcel Dekker, (1998).
\bibitem{ishi2016}
{T. Ishii, H. Ono and Y. Uno:}
\newblock (Total) Vector domination for graphs with bounded branchwidth.
\newblock {Discrete Appl. Math.,} {\bf 207}, 80--89, (2016).
\bibitem{klasing2004hardness}
{ R. Klasing and C. Laforest:}
\newblock Hardness results and approximation algorithms of $k$-tuple domination
  in graphs.
\newblock {Inform. Process. Lett.}, {\bf 89},  75--83, (2004).

\bibitem{minty1980maximal}
{  G. J. Minty:}
\newblock On maximal independent sets of vertices in claw-free graphs.
\newblock { J. Comb. Theory. B.}, {\bf 28} 284--304, (1980).

\bibitem{Oum2006}
{  S.-i. Oum and P. Seymour:}
\newblock Approximating clique-width and branch-width.
\newblock { J. Comb. Theory. B.}, {\bf 96} 514--528, (2006).

\end{thebibliography}
\end{document}